\documentclass[10pt, conference, letterpaper]{IEEEtran}
\IEEEoverridecommandlockouts

\usepackage{cite}
\usepackage{amsmath,amsfonts,amsthm}
\usepackage{algorithm}
\usepackage{algorithmic}
\usepackage{graphicx}
\usepackage{textcomp}
\usepackage{xcolor}
\usepackage{tabularx}
\usepackage{subfig}
\usepackage{comment}
\def\BibTeX{{\rm B\kern-.05em{\sc i\kern-.025em b}\kern-.08em
    T\kern-.1667em\lower.7ex\hbox{E}\kern-.125emX}}

\usepackage{bm}
\newtheorem{theorem}{Theorem}

\usepackage{amsmath}

\begin{document}



\title{Exploiting Storage for Computing: Computation Reuse in Collaborative Edge Computing}

\author{\IEEEauthorblockN{Xingqiu He\IEEEauthorrefmark{1},
        Chaoqun You\IEEEauthorrefmark{2},
        Tony Q. S. Quek\IEEEauthorrefmark{3}}
        \IEEEauthorblockA{\IEEEauthorrefmark{1}\IEEEauthorrefmark{2}Fudan University, \IEEEauthorrefmark{3}Singapore University of Technology and Design \\
        E-mail: \IEEEauthorrefmark{1}hexqiu@gmail.com,
        \IEEEauthorrefmark{2}chaoqunyou@gmail.com,
        \IEEEauthorrefmark{3}tonyquek@sutd.edu.sg
        }
    }

\maketitle

\begin{abstract}
Collaborative Edge Computing (CEC) is a new edge computing paradigm that enables neighboring edge servers to share computational resources with each other. 
Although CEC can enhance the utilization of computational resources, it still suffers from resource waste.
The primary reason is that end-users from the same area are likely to offload similar tasks to edge servers, thereby leading to duplicate computations. 
To improve system efficiency, the computation results of previously executed tasks can be cached and then reused by subsequent tasks. 
However, most existing computation reuse algorithms only consider one edge server, which significantly limits the effectiveness of computation reuse.
To address this issue, this paper applies computation reuse in CEC networks to exploit the collaboration among edge servers.
We formulate an optimization problem that aims to minimize the overall task response time and decompose it into a caching subproblem and a scheduling subproblem.
By analyzing the properties of optimal solutions, we show that the optimal caching decisions can be efficiently searched using the bisection method.
For the scheduling subproblem, we utilize projected gradient descent and backtracking to find a local minimum.
Numerical results show that our algorithm significantly reduces the response time in various situations.
\end{abstract}


\section{Introduction}
In the past decade, the rapid evolution of machine learning and Internet of Things has driven the 
development of many computation-intensive applications, such as AR/VR~\cite{qiao2019web} and mobile gaming~\cite{lin2016cloudfog}.
To fulfill the stringent Quality of Service (QoS) requirements of these applications,
edge computing is regarded as a promising computing paradigm that provides computation-offloading services in proximity to end-users.
Nevertheless, due to limited resources at edge servers,
an edge computing system may suffer from severe performance degradation when dealing with bursty workloads~\cite{he2021providing}. 
To address this issue, Collaborative Edge Computing (CEC)~\cite{tran2017collaborative}
has been proposed to alleviate load imbalance by offloading computation tasks from busy edge servers to idle ones.
Extensive studies~\cite{xiao2017qoe, chen2018computation, li2019learning, lyu2018distributed} have shown that CEC effectively improves the utilization of computational resources and thus
greatly enhances the system performance.
  
Despite the significant improvement CEC brings to resource utilization, computational resources in CEC networks are still grossly wasted due to duplicate computations.
As shown in~\cite{al2022promise}, clustered end-users are likely to offload similar computation tasks to edge servers due to temporal, spatial, or semantic correlations among input data. 
For example, end-users in AR/VR may request the rendering of the same frame if they have similar poses~\cite{meng2020coterie},
and intelligent surveillance systems often need to analyze similar images captured by the camera in a short period of time~\cite{al2022promise}. 
Based on this observation, \textit{computation reuse}~\cite{guo2018potluck, guo2018foggycache} has been proposed to further reduce resource waste.
The main idea is, the computation results of previously processed tasks can be cached and reused if subsequent tasks have similar input. 
This mechanism not only avoids duplicate computations, but also reduces the system response time, as tasks with cached results can be accomplished immediately without being processed from scratch.

Recently, extensive studies have been conducted to explore computation reuse at a single edge server~\cite{lee2019case, bellal2021coxnet, al2022promise, nour2021whispering}.
However, the benefits of computation reuse can only be fully realized through the collaboration of multiple edge servers.
The main reason is, with collaboration, the computation results cached at one edge server can be reused by tasks offloaded to other edge servers.
This not only improves the cache hit rate for each cached result, but also increases the total number of cached results in the system, as we no longer need to redundantly cache the same results at different edge servers.

Despite the aforementioned benefits, implementing computation reuse in CEC
networks is particularly challenging due to the following two reasons.
First, compared to the computation reuse at a single edge server, the mutual influence among caching decisions at different edge
servers is hard to quantify.
Second, compared to the conventional CEC networks, the introduction of computation reuse changes the distribution of task processing
time, thereby making the caching and scheduling decisions deeply intertwined and difficult to optimize.


In this paper, we propose a computation reuse algorithm for CEC networks.
Our objective is to minimize the weighted response time by jointly optimizing the computation result caching decisions, the cache-searching strategy, the workload scheduling policy, and the computational power allocation scheme.
The formulated problem is a mixed integer programming problem and is proved to be NP-hard even for an extremely simplified special case.
In order to solve it, we apply alternating minimization to decompose the original problem into a caching subproblem and a scheduling subproblem.
Based on the necessary conditions of optimal solutions and the special properties of subproblems,
we propose two efficient approximate algorithms that converge to local minima in polynomial time.

The main contributions are summarized as follows.
\begin{itemize}
    \item We propose to implement computation reuse in CEC networks and formulate the corresponding response time minimization problem.
        To reduce the difficulty of algorithm design, we decompose the formulated optimization into 
        a caching subproblem and a scheduling subproblem.
    \item For the caching subproblem, we introduce the concept of storage efficiency and prove that it 
        satisfies certain conditions at optimal solutions.
        Based on this property, we further identify the monotonicity of storage efficiency and propose
        an efficient approximate algorithm based on the bisection search.
        Theoretical analysis guarantees the obtained solution is near-optimal in practical situations.
    \item For the scheduling subproblem, we apply projected gradient descent to find a local minimum iteratively.
        A key step in projected gradient descent is projecting the obtained solution to the feasible region.
        We use variable substitution to transform the feasible region to a probability simplex so that such
        projection can be implemented efficiently.
        We also ameliorate the backtracking line search to accelerate the convergence process and guarantee the stability of obtained solutions.
    \item We conduct extensive simulations to validate the effectiveness of the proposed algorithm.
        Our algorithm is compared with three benchmarks and
        numerical results show that our algorithm significantly improves the system performance under various settings.
\end{itemize}

The rest of the paper is organized as follows. 
In Section \ref{section:related_work}, we review related works.
Section \ref{section:system_model} introduces the system model and formulates the corresponding optimization problem.
Section \ref{section:algorithm_design} describes the algorithm design and related analysis.
The numerical results and conclusions are presented in Section \ref{section:simulation} and Section \ref{section:conclusion}.

\section{Related Work} \label{section:related_work}
\textbf{Collaborative Edge Computing.}
As an effective load-balancing approach, the scheduling algorithms in CEC networks have been extensively studied from different perspectives.
The authors in \cite{xiao2017qoe} developed a distributed optimization framework that achieves optimal tradeoff between the user
experience and power efficiency of edge servers.
The work in \cite{chen2018computation} proposed a novel online peer offloading algorithm to optimize the long-term
latency under the time-average energy consumption constraint.
A distributed optimization for cost-effective offloading decisions was investigated in \cite{lyu2018distributed}.
The authors in \cite{li2019learning} and \cite{he2021providing} consider scheduling algorithms that aim to meet the deadlines of tasks.
However, these algorithms only optimize the allocation of computational resources without considering the caching problem.

\textbf{Collaborative Caching.}
To fulfill the ever-increasing demand for content delivery, the collaborative caching among edge servers is studied in 
\cite{saputra2019novel, ostovari2016efficient, wang2017multi, zhang2017cooperative}.
These works usually minimize the content-access delay or cost of content providers based on 
content popularity \cite{ostovari2016efficient, wang2017multi}, spectrum efficiency \cite{zhang2017cooperative}, network topology and user' demands\cite{saputra2019novel}. 
In addition to content caching, there are also plenty of works 
\cite{ma2020cooperative, feng2021collaborative, xu2021energy, zhong2021cooperative}
considering collaborative service caching and computation offloading.
They assume tasks must be processed by edge servers that have cached related data so a joint optimization of 
caching and offloading is required.
For example, the authors in \cite{ma2020cooperative} developed an iterative algorithm to minimize the response time
by jointly optimizing the service caching decisions and workload scheduling policies.
However, the service caching decisions do not influence the service time of tasks, hence the problems formulated in the above works
are usually much simpler than the problem considered in this paper.

\textbf{Computation Reuse.}
The concept of computation reuse is first used to deduplicate redundant computation across multiple 
applications \cite{guo2018potluck} and multiple devices \cite{guo2018foggycache}.
After that, this technique is extensively explored in edge computing to improve the system efficiency 
\cite{lee2019case, bellal2021coxnet, al2022promise}.
To realize computation reuse among multiple edge servers, the authors in \cite{al2022reservoir} proposed a framework
that runs on top of named-data networking \cite{zhang2014named}.
However, the work focuses on communication among edge servers and does not optimize the caching and scheduling decisions.
The work in \cite{nour2021whispering} is the first one that designs the scheduling algorithm in a principled way,
but the modeling is simplified (e.g. the caching decisions are not considered) and thus the algorithm's practicality is limited.
To the best of our knowledge, this paper is the first work that considers computation reuse in CEC networks.

\section{System Model and Problem Formulation} \label{section:system_model}

In this section, we describe the considered CEC system with computation reuse and formulate the corresponding response time minimization problem.
Complexity analysis shows that the formulated problem is NP-hard even for a significantly simplified special case.

\subsection{Computation Reuse Model}
\begin{figure}[t]
\centering
\includegraphics[width=0.4\textwidth]{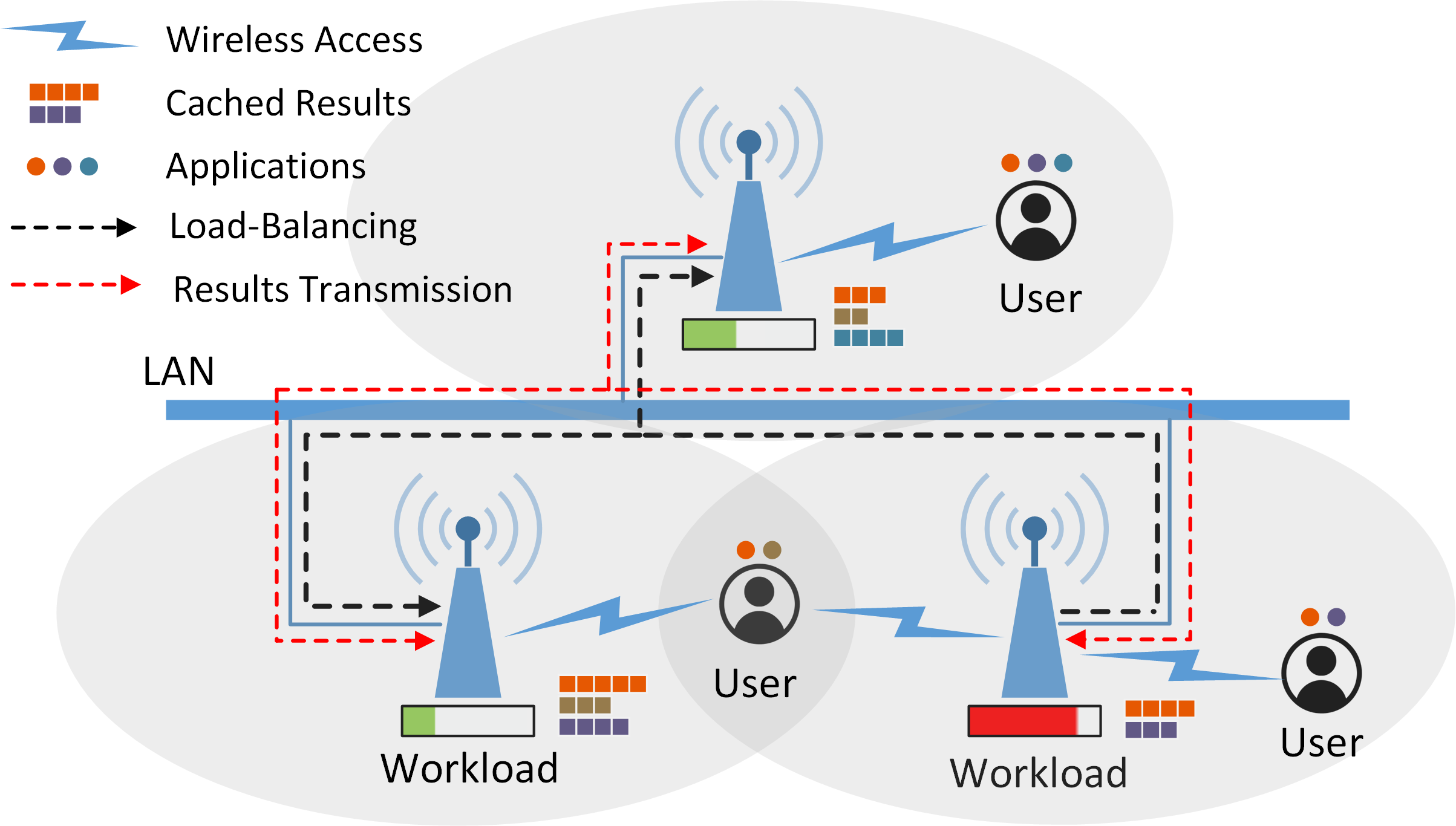}
\caption{A simple example of the considered system model.}
\label{fig:system}
\end{figure}
As shown in Fig.\ref{fig:system}, we consider a CEC network with $N$ BSs indexed by $\mathcal{N} = \{1, 2, \dots, N\}$. 
Each BS $n\in\mathcal{N}$ is co-located with an edge server with computational capacity $C_n^c$ (in CPU frequency) and storage capacity
\footnote{In current practice, the computation results are usually stored in RAM so $C^s_n$ refers to the RAM space. But our system model and proposed algorithm also apply to disk storage situations.} 
$C_n^s$ (in GB), respectively.
Suppose there are $A$ different applications indexed by $\mathcal{A} = \{1, 2,\dots, A \}$ running in the CEC network. 
The end-users running these applications may offload computation tasks to nearby BSs with different input data.
Taking an AR/VR application for example, the input data of each task indicate the user's pose, comprised of 3D position and 3D orientation.
The output of the task is the rendered frame to be displayed on the user's head-mounted displays.

According to the results in \cite{al2022reservoir}, 
for a given application, computation tasks with sufficiently similar input data usually have the same outputs,
so the results of previous tasks can be reused by subsequent tasks. 
Specifically, for application $a\in\mathcal{A}$, suppose we have identified a set of typical inputs $\mathcal{K}^a = \{1,2,\dots, K^a\}$ from its historical data. 
If a new task of application $a$ is offloaded from end-users, we can first find the closest typical input using fast and space-efficient mechanisms, such as Locality Sensitive Hashing (LSH)~\cite{datar2004locality} and Feature Hashing (FH)~\cite{weinberger2009feature}. 
After that, we compute the similarity between the two inputs according to the chosen metric (e.g., structural similarity \cite{wang2004image} or cosine similarity \cite{singhal2001modern}). 
If the similarity exceeds a certain threshold and the output of the corresponding typical input has been cached, the computation result can be directly returned without processing the task from scratch.

Although computation reuse saves a significant amount of computational resources when the target result is cached, 
it also incurs extra computational costs due to operations like hashing and similarity calculation. 
Therefore, in some cases, it may be beneficial to directly process tasks without searching for the cached results, 
especially when the probability of successfully finding the target result is low. 
To allow for a flexible cache-searching scheme, we introduce a binary variable $y_n^a \in \{0, 1\}$ to indicate whether 
we will search for a cached result for computation tasks from application $a$ at BS $n$. 
If $y^a_n=0$, then all tasks of application $a$ will be computed from scratch at BS $n$, while $y^a_n=1$ indicates otherwise.

Clearly, the optimal value of $y^a_n$ depends on how often we successfully find the target result.
Suppose we are processing a task at BS $n$, then we can define the \textit{local hit rate} $P_{lhr}^{a,n}$ and \textit{neighboring hit rate} $P^{a,n}_{nhr}$ as the probabilities that the desired computation result is cached at BS $n$ and the neighboring BSs, respectively.
For any typical input $k\in\mathcal{K}^a$, let $x^a_{nk}\in\{0,1\}$ indicate whether the output of $k$ is cached at BS $n$,
and let $p_k^a$ denote the probability that a task from application $a$ is sufficiently similar to $k$.
Then the local and neighboring hit rate can be expressed as $P_{lhr}^{a,n} = \sum_{k\in\mathcal{K}^a} x_{nk}^a p_k^a$
and $P^{a,n}_{nhr} = \sum_{k\in\mathcal{K}^a} p^a_k(1-x^a_{nk})I_{>0}(\sum_{m\in\mathcal{N}, m\ne n}x^a_{mk})$,
where $I_{>0}(\cdot)$ is the indicator function that returns $1$ if the input value is positive.
For convenience, we define $P^{a,n}_{hr} = P^{a,n}_{lhr} + P^{a,n}_{nhr}$ as the aggregate hit rate.
Apparently, $P^{a,n}_{hr}$ is the hit rate of application $a$ in the whole system and remains the same for different BS $n$.
Therefore, we can drop the superscript $n$ and use $P^{a}_{hr}$ afterward.

Although it is desirable to enhance the hit rate as much as possible, the number of cached results is constrained by the storage capacities of BSs.
Suppose the data size of the computation result corresponding to typical input $k$ is $s_k^a$, then the cached data at BS $n$ should not exceed its storage capacity, which can be expressed as
\begin{equation}
    \sum_{a\in\mathcal{A}} \sum_{k\in\mathcal{K}^a} x^a_{nk}s^a_k \leq C^s_n, \quad \forall n\in\mathcal{N}. \label{cons:storage}
\end{equation}

\subsection{Task Scheduling Model}
In this subsection, we introduce the decision variables related to task scheduling and derive the expression of response time.
Let $D^a$ be the average response time of application $a$. Generally, $D^a$ consists of two parts: the computation delay that is caused by processing the computation tasks and the transmission delay that is caused by transmitting tasks and computation results among BSs. In the following, we will provide the formulations of computation delay and transmission delay, respectively.

Assume that the task arrival at BS $n$ follows a Poisson process with expected rate $R_n^a$. For convenience, we define the total arrival rate of all BSs in the system as $R^a = \sum_{n\in\mathcal{N}} R^a_n$. 
In CEC, the computation workload can be redistributed among collaborative BSs.
Therefore, we introduce $\lambda_n^a \in [0,1]$ to denote the fraction of application $a$'s workload distributed to BS $n$. Apparently, for any application $a$, the sum of its workload distribution across all BSs should be equal to $1$, hence we have
\begin{equation}
    \sum_{n\in\mathcal{N}} \lambda^a_n = 1, \quad \forall a\in\mathcal{A}. \label{cons:scheduling}
\end{equation}

As discussed in the previous subsection, the processing procedures of tasks are different depending on the value of $y^a_n$.
Therefore, we will derive the expression of computation delay case by case.
Let $f^a_n$ be the computational capability (in CPU frequency) allocated to application $a$ at BS $n$
and $w^a$ be the average workload (in CPU cycles) of tasks from application $a$. 
As in~\cite{xu2018joint, ma2020cooperative}, we assume the workload of tasks follows an exponential distribution.
When $y_n^a = 0$, the computation process satisfies the M/M/1 queueing model, so the average computation delay without cache-searching is
\begin{equation}
    D^{a,0}_n = \frac{1}{\mu^{a,0}_n - \lambda^a_n R^a}, \label{eq:delay_no_cache}
\end{equation}
where $\mu^{a,0}_n = f^a_n/w^a$ is the average processing rate with respect to the number of tasks.
To guarantee the stability of the queue, we must ensure
\begin{equation}
    \lambda^a_n R^a < \mu^{a,0}_n, \quad \forall a\in\mathcal{A}, \forall n\in\mathcal{N}. \label{cons:stability_0}
\end{equation}

When $y^a_n=1$, we will first search the cached results before processing. 
If the task's computation result is cached and reused, then the computation cost of the task is aroused by the searching process, which is denoted by $w^s$ (in CPU cycles). 
Otherwise, the task still needs to be processed from scratch.
Consequently, the Cumulative Distribution Function (CDF) of the total computation cost $w$ can be derived as
\begin{equation*}
    F(w) = 
    \begin{cases}
        1 - (1-P^{a}_{hr})e^{-(w-w^s)/w^a}, & w \geq w^s \\
        0, & w < w^s
    \end{cases},
\end{equation*}
where $F(w)$ jumps from $0$ to $P^{a}_{hr}$ at $w^s$, and gradually increases to $1$ thereafter.

According to the M/G/1 queueing model, the average computation delay with cache-searching is
\begin{align}
    D^{a,1}_n = &\frac{1}{\mu^{a,1}_n} + \frac{\lambda^a_nR^a}{2\mu^{a,1}_n(\mu^{a,1}_n-\lambda^a_nR^a)} \notag \\
    &+ \frac{(1-P^{a}_{hr})(1+P^{a}_{hr})\lambda^a_nR^a\mu^{a,1}_n}{2(\mu^{a,1}_n-\lambda^a_nR^a)(\mu^{a,0}_n)^2}
    \label{delay:cache}
\end{align}
where $\mu^{a,1}_n = f^a_n/(w^s+(1-P^{a}_{hr})w^a)$.
To guarantee the stability of the queue, we must have
\begin{equation}
    \lambda^a_n R^a < \mu^{a,1}_n, \quad \forall a\in\mathcal{A}, \forall n\in\mathcal{N}.
    \label{cons:stability_1}
\end{equation}
Combining \eqref{cons:stability_0} and \eqref{cons:stability_1} yields
\begin{equation}
    \lambda^a_n R^a - (1-y^a_n)\mu^{a,0}_n - y^a_n\mu^{a,1}_n < 0, \quad \forall a\in\mathcal{A}, \forall n\in\mathcal{N}.
    \label{cons:stability}
\end{equation}

Notice that if the target result is cached in other BSs, it must be transmitted to the BS $n$ for reuse. As in \cite{ma2020cooperative}, the transmission delay from neighboring BSs to BS $n$ can be assumed to be a constant $D^t_n$.
Hence, the average processing delay of application $a$ at BS $n$ is given by
\begin{equation}
    D^a_n = (1-y^a_n)D^{a,0}_n + y^a_n(D^{a,1}_n + P^{a,n}_{nhr}D^t_n).
    \label{eq:D^a_n}
\end{equation}

In addition to the processing delay $D_n^a$, the response time $D^a$ of application $a$ also includes the transmission delay caused by transmitting tasks and computation results among BSs.
If $\lambda^a_nR^a - R^a_n > 0$, then there are extra tasks offloaded from other BSs to BS $n$ and the induced transmission delay is
$D^t_n(\lambda^a_nR^a - R^a_n)/R^a$.
Otherwise, if $\lambda^a_nR^a - R^a_n < 0$, then part of BS $n$'s tasks will be offloaded to the other BSs.
The transmission delay induced by returning computation results to BS $n$ is $D^t_n(R^a_n - \lambda^a_n R^a)/R^a$.
Based on the above discussion, the expression of $D^a$ is given as
\begin{equation*}
    D^a = \sum_{n\in\mathcal{N}} \left[ \lambda^a_nD^a_n + \frac{|\lambda^a_n R^a - R^a_n|}{R^a}D^t_n \right].
\end{equation*}

In practice, it may be beneficial to offload part of tasks to the cloud when the local BSs are overloaded.
This situation can be accommodated in our model by regarding the remote cloud as a neighboring BS with sufficient computational power and large transmission delays.
If the bandwidth of the core network is limited, one can also treat the transmission process as an M/M/1 queueing model (as in \cite{ma2020cooperative}).
The resulting transmission delay is a function of the offloaded workload, which takes a similar form of \eqref{eq:delay_no_cache},
and hence brings no extra complexity to the problem-solving.
Therefore, to keep the description concise and focus on the key parts of the problem, we do not explicitly consider offloading to the cloud in our model.

\subsection{Problem Formulation}
This paper aims to minimize the weighted response time of different applications by jointly optimizing the computation result caching decisions $x^a_{nk}$,
the workload scheduling policy $\lambda^a_n$, the cache-searching strategy $y^a_n$, and the computational capability allocation $f^a_n$.
Let $\phi^a$ be the weight of application $a$ that reflects its priority, then the corresponding optimization problem can be formulated as:
\begin{align}
    \min_{\bm{x},\bm{y},\bm{\lambda},\bm{f}}\quad & \sum_{a\in\mathcal{A}} \phi^a D^a \label{lmp} \\
    s.t.\quad & \eqref{cons:storage}, \eqref{cons:scheduling}, \eqref{cons:stability} \quad \tag{\ref{lmp}{a}} \label{lmp:before_cons} \\
    & \sum_{a\in\mathcal{A}} f^a_n \leq C^c_n, \quad \forall n\in\mathcal{N} \tag{\ref{lmp}{b}} \label{lmp:computing} \\
    & x^a_{nk}\in\{0,1\}, \quad \forall a\in\mathcal{A}, \forall n\in\mathcal{N}, \forall k\in\mathcal{K}^a \tag{\ref{lmp}{c}} \label{lmp:variable_x} \\
    & y^a_n\in\{0,1\}, \lambda^a_n, f^a_n \geq 0, \quad \forall a\in\mathcal{A}, \forall n\in\mathcal{N} \tag{\ref{lmp}{d}} \label{lmp:variable_ylf}
\end{align}
where $\bm{x}, \bm{y}, \bm{\lambda}, \bm{f}$ are vectors of corresponding decision variables.
Constraint \eqref{lmp:before_cons} contains the caching and scheduling constraints described above.
Constraint \eqref{lmp:computing} requires that the computational power allocated to applications cannot exceed the BS's computational capacity.
Constraint \eqref{lmp:variable_x} and \eqref{lmp:variable_ylf} define the feasible regions of control variables.

The formulated problem \eqref{lmp} is a mixed integer non-linear programming problem and has no polynomial-time solution in general.
In fact, we can show that the problem \eqref{lmp} is NP-hard even for a significantly simplified special case, as stated in the following theorem:
\begin{theorem}
    The problem \eqref{lmp} is NP-hard even if there is only one BS and one application in the system.
    \label{theorem:np}
\end{theorem}
\begin{proof}
   Please see Appendix \ref{appendix:1}.
\end{proof}

Theorem \ref{theorem:np} demonstrates that it is computationally prohibitive to compute the optimal caching decisions even for a single BS and a single application.
In realistic situations, we also need to consider the correlation between caching decisions at different BSs,
the workload scheduling policy among BSs, and the allocation of computational power to each application.
All of these aspects substantially complicate the formulated problem \eqref{lmp} and make it very challenging to solve.

\section{Algorithm Design and Theoretical Analysis} \label{section:algorithm_design}
In this section, we first relax constraint \eqref{cons:stability} and show that the optimal value of $y^a_n$ can be obtained
by comparing $D^{a,0}_n$ and $D^{a,1}_n$.
After that, the relaxed problem is decomposed into two subproblems by using the Alternating Minimization (AM) 
method.\footnote{Also known as the Block Coordinate Descent method.}
Based on the necessary conditions of the optimal solutions, two efficient algorithms are proposed to solve the resulting subproblems.

\subsection{Problem Relaxation and Decomposition}

For convenience, we temporarily relax constraint \eqref{cons:stability} and will re-consider it in section \ref{subsection:scheduling}.
To optimize the weighted response time, we should choose the cache-searching strategy that minimizes the processing delay
for each BS and each application.
According to \eqref{eq:D^a_n}, it is obvious that the optimal value of $y^a_n$ is $y^a_n=1$ when 
$D^{a,0}_n > D^{a,1}_n + P^{a,n}_{nhr}D^t_n$ and $y^a_n=0$ otherwise.
Therefore, $\bm{y}$ is automatically determined once $\bm{x}, \bm{\lambda}$ and $\bm{f}$ are resolved.
Hence, we will not consider $\bm{y}$ explicitly in the following analysis.


In problem \eqref{lmp}, different decision variables are coupled in the non-convex objective function, thereby making it difficult for the variables to be optimized simultaneously. 
To reduce the difficulty of algorithm design, we resort to the AM method to decouple the joint optimization problem \eqref{lmp} into a sequence of subproblems, each of which includes a subset of variables. 
In our problem, the decision variables are partitioned into two groups:
the caching variables $x^a_{nk}$ and the scheduling variables $\lambda^a_n, f^a_n$.
The corresponding caching subproblem aims to find the optimal caching decisions with the given scheduling policy, which is given as follows
\begin{align}
    \min_{\bm{x}}\quad & \sum_{a\in\mathcal{A}} \phi^a D^a \label{subproblem:caching} \\
    s.t.\quad & \eqref{cons:storage} \mbox{ and } \eqref{lmp:variable_x}. \quad \notag
\end{align}
The scheduling subproblem aims to optimize the allocation of workload and CPU frequency with fixed caching decisions, which is given as follows
\begin{align}
    \min_{\bm{\lambda},\bm{f}}\quad & \sum_{a\in\mathcal{A}} \phi^a D^a \label{subproblem:scheduling} \\
    s.t.\quad & \eqref{cons:scheduling}, \eqref{lmp:computing}, \eqref{lmp:variable_ylf}. \quad \notag
\end{align}
According to the theoretical results in \cite{bezdek2003convergence}, if we optimize \eqref{subproblem:caching}
and \eqref{subproblem:scheduling} alternately, the obtained decision variables will eventually converge to local minima.
Unfortunately, due to the difficulty of the original problem, both subproblems are still hard to solve.
Specifically, the caching subproblem is also NP-hard according to the proof of Theorem \ref{theorem:np}.
In addition, the objective function in the scheduling subproblem is non-convex and thus has no efficient solutions in general.
To reduce the algorithm's time complexity, in the next two subsections, we will propose two polynomial-time approximation algorithms
based on the necessary conditions of optimal solutions.

\subsection{Caching Subproblem} \label{subsection:caching}

The mutual influence between caching decisions at different BSs is very complicated, so we start by focusing on only one BS.
Let $\bm{x}_n$ and $\bm{x}_{-n}$ be the caching decisions at BS $n$ and the rest BSs, respectively.
Our goal is to optimize $\bm{x}_n$ for any given $\bm{x}_{-n}$.
For convenience, we temporarily relax the binary constraint of $x^a_{nk}$ and assume it is a real number in $[0, 1]$. The relaxed caching subproblem at BS $n$ is called problem $P^c_n$.
We will reconstruct a binary caching decision after obtaining the optimal solution of $P^c_n$.

According to the expressions of $D^a_n$ and $D^{a,1}_n$, the influence of caching decisions on the objective function is through the hit rate.
One can easily verify that the objective function is non-increasing with respect to the local hit rate $P^{a,n}_{lhr}$,
hence it is also non-increasing with respect to the caching decisions $x^a_{nk}$.
Therefore, to minimize the objective function, one should increase the value of $x^a_{nk}$ as much as possible.
However, in our problem, the caching decisions at the same BS are constrained by the storage space.
To maximize the utilization efficiency, the storage space should be allocated to the $x^a_{nk}$ that contributes most to the reduction of the value of the objective function.
Inspired by this idea, we use $\sigma^{a}_{nk} = x^{a}_{nk} s^a_k$ to denote the storage space allocated to $x^a_{nk}$ and 
define the storage efficiency of $x^a_{nk}$ as
\begin{align}
    \epsilon(x^a_{nk}) &= \frac{\partial \sum_{a'\in\mathcal{A}} \phi^{a'} D^{a'}}{\partial \sigma^{a}_{nk}} \notag \\
    &= \frac{\partial \sum_{a'\in\mathcal{A}} \sum_{n'\in\mathcal{N}} \phi^{a'} \lambda^{a'}_{n'} y^{a'}_{n'} (D^{a',1}_{n'} + P^{a',n'}_{nhr} D^t_{n'})}
    {\partial \sigma^a_{nk}} \notag \\
    &= \sum_{n'\in\mathcal{N}} \phi^a \lambda^a_{n'} y^a_{n'} \left( \frac{\partial D^{a,1}_{n'}}{\partial P^{a}_{hr}} 
    \frac{\partial P^a_{hr}}{\partial x^a_{nk}s^a_k} + D^t_{n'} \frac{\partial P^{a,n'}_{nhr}}{\partial x^a_{nk}s^a_k}  \right). \label{eq:x_a_nk}
\end{align}
Let $I^{a}_{-n,k} = I_{>0}(\sum_{m\in\mathcal{N},m\ne n} x^a_{mk})$ be the binary indicator indicating whether the result of the 
$k$-th typical input in $\mathcal{K}^a$ is cached outside BS $n$.
Then according to the definition of $P^a_{hr}$ and $P^{a,n}_{nhr}$ we have 
\begin{align}
    \frac{\partial P^a_{hr}}{\partial x^a_{nk}} &= p^a_k(1-I^a_{-n,k}), \label{eq:P_a_hr_1} \\
    \frac{\partial P^{a,n'}_{nhr}}{\partial x^a_{nk}} &= 
    \begin{cases}
        p^a_k(1 - x^a_{n'k})(1 - I^a_{-n,k}) & n' \ne n, \label{eq:P_a_nhr_2} \\
        -p^a_k I^a_{-n',k} & n' = n.
    \end{cases}
\end{align}
Since $x^a_{n'k} = 1$ implies $I^a_{-n,k} = 1$ when $n' \ne n$, the term $p^a_k(1 - x^a_{n'k})(1 - I^a_{-n,k})$ can be simplified to
$p^a_k(1 - I^a_{-n,k})$. Substituting \eqref{eq:P_a_hr_1} and \eqref{eq:P_a_nhr_2} into \eqref{eq:x_a_nk} implies
\begin{align}
    \epsilon(x^a_{nk}) =&
    \frac{p^a_k(1-I^a_{-n,k})}{s^a_k} 
    \Biggl[ \phi^a \lambda^a_{n} y^a_{n} \frac{\partial D^{a,1}_{n}}{\partial P^a_{hr}} \notag \\
    &+\sum_{n'\in\mathcal{N}, n'\ne n} \phi^a \lambda^a_{n'} y^a_{n'} 
    \left( \frac{\partial D^{a,1}_{n'}}{\partial P^a_{hr}} + D^t_{n'} \right) \Biggr] \notag \\
    &- \frac{p^a_k I^a_{-n,k}}{s^a_k} \phi^a \lambda^a_n y^a_n D^t_n, \label{eq:storage_efficiency_2}
\end{align}
where
\begin{align*}
    \frac{\partial D^{a,1}_n}{\partial P^a_{hr}} = &- \frac{w^a}{f^a_n} 
    - \frac{(\lambda^a_n R^a)^2 w^a ((1-P^a_{hr})w^a + w^s)^2}{2 f^a_n (f^a_n - \lambda^a_n R^a ((1-P^a_{hr}) w^a + w^s))^2} \\
    &- \frac{f^a_n (\lambda^a_n R^a)^2 (1-P^a_{hr})(1+P^a_{hr}) w^a}{2 (\mu^{a,0}_n)^2 (f^a_n - \lambda^a_n R^a ((1-P^a_{hr}) w^a + w^s))^2} \\
    &- \frac{f^a_n \lambda^a_n R^a P^a_{hr}}{(\mu^{a,0}_n)^2 (f^a_n - \lambda^a_n R^a ((1-P^a_{hr}) w^a + w^s))} \\
    &- \frac{\lambda^a_n R^a w^a ((1-P^a_{hr})w^a + w^s)}{f^a_n (f^a_n - \lambda^a_n R^a ((1-P^a_{hr}) w^a + w^s))}.
\end{align*}

As discussed above, to minimize the system response time, we should allocate the storage space to those typical inputs $k$
with the smallest negative storage efficiency.
Based on this idea, we can formally prove that the optimal solution of problem $P^c_n$ should satisfy the following property:
\begin{theorem}
    If $\bm{x}^*_n = (x^{a,*}_{nk})_{a\in\mathcal{A}, k\in\mathcal{K}^a}$ denote the optimal solution to $P^c_n$, then 
    there exists a constant $B^*$ such that for all application $a$ and typical input $k$, we have 

    (1) $\sum_{a\in\mathcal{A}, k\in\mathcal{K}^a} x^{a,*}_{nk} s^a_k = C^s_n$;

    (2) if $x^{a,*}_{nk}$ is fractional, i.e. $x^{a,*}_{nk} \in (0,1)$, then $\epsilon(x^{a,*}_{nk}) = B^*$;

    (3) if $x^{a,*}_{nk} = 1$, then $\epsilon(x^{a,*}_{nk}) \leq B^*$;

    (4) if $x^{a,*}_{nk} = 0$, then $\epsilon(x^{a,*}_{nk}) \geq B^*$.
    \label{theorem:optimality}
\end{theorem}
\begin{proof}
    Please see Appendix \ref{appendix:2}.
\end{proof}
By calculating the derivative of $\epsilon(x^a_{nk})$ one can verify that $\epsilon(x^a_{nk})$ is an increasing function
so there is a one-to-one correspondence\footnote{In fact, $\epsilon(x^a_{nk})$ may not be strictly increasing so there may exist
    multiple $x^a_{nk}$ with the same $\epsilon(x^a_{nk})$. However, these $x^a_{nk}$ also result in the same objective
value so we do not need to discriminate them.} between $\epsilon(x^a_{nk})$ and $x^a_{nk}$.
Based on this result, we can further prove that for each given constant $B$, there is only one $\bm{x}_n$ which satisfies the last three 
statements in Theorem \ref{theorem:optimality}.
For convenience, we use function $g$ to denote the correspondence between $B$ and $\bm{x}_n$, i.e. $\bm{x}_n = g(B)$.
Since $\epsilon(x^a_{nk})$ is increasing, then by the second statement of Theorem \ref{theorem:optimality},
a larger $B$ corresponds to a larger $x^a_{nk}$.
Therefore, the totally allocated storage space $\sum_{a\in\mathcal{A}, k\in\mathcal{K}^a} x^a_{nk} s^a_k$ is also increasing
with respect to $B$.
As a result, we can use the bisection search to find the optimal $B^*$ that ensures 
$\sum_{a\in\mathcal{A}, k\in\mathcal{K}^a} x^a_{nk} s^a_k = C^s_n$.
After that, the optimal solution can be immediately obtained by $\bm{x}^*_n = g(B^*)$ according to Theorem \ref{theorem:optimality}. 

Since we have obtained $B^*$, the next step is to design an algorithm that implements $g(B)$.
The main difficulty comes from the fact that $\epsilon(x^a_{nk})$ is a function of $P^a_{nk}$, which depends on
the value of $x^a_{nk'}$ for all $k'\in\mathcal{K}^a$.
Therefore, the values of different $x^a_{nk}$ are mutually influenced and must be solved jointly.
Although this problem can be very complicated in general, we will show that $\epsilon(x^a_{nk})$ exhibits
a desirable property that substantially facilitates the solving efficiency.


According to \eqref{eq:storage_efficiency_2}, for any given BS $n$ and application $a$, the storage efficiency $\epsilon(x^a_{nk})$ 
for different typical inputs $k$ can be classified into two classes depending on the value of $I^a_{-n,k}$.
If $I^a_{-n,k} = 1$, the first term in \eqref{eq:storage_efficiency_2} becomes zero and $\epsilon(x^a_{nk})$ is reduced to a constant.
Similarly, if $I^a_{-n,k} = 0$, the second term becomes zero and $\epsilon(x^a_{nk})$ equals to $p^a_k/s^a_k$ multiplying
a term that is invariant under different $k$.
In this case, the order of $\epsilon(x^a_{nk})$ for different typical inputs $k$ is determined by $p^a_k/s^a_k$ and is 
independent of $x^a_{nk}$.
Based on such observation, we can solve $\bm{x}_n = g(B)$ by the following process.

Since the caching decisions of different applications are independent, we only determine the values of $x^a_{nk}$ for
one application $a$ each time.
Let $\mathcal{K}^a_1$ and $\mathcal{K}^a_0$ denote the two classes corresponding to $I^{a}_{-n,k} = 1$ and $I^{a}_{-n,k} = 0$, respectively.
If $k\in\mathcal{K}^a_1$, then $\epsilon(x^a_{nk})$ is a constant hence $x^a_{nk} = 1$ when $\epsilon(x^a_{nk}) \leq B$ 
and $x^a_{nk} = 0$ otherwise.
For $k\in\mathcal{K}^a_0$, we first sort them in descending order according to the value of $p^a_k/s^a_k$,
denoted as $\pi_1,\dots,\pi_{K^a_0}$, where $K^a_0 = | \mathcal{K}^a_0 |$ is the number of elements in $\mathcal{K}^a_0$.
Since $\epsilon(x^a_{nk})$ is negative, then by the above analysis we have 
$\epsilon(x^a_{n\pi_1}) \leq \epsilon(x^a_{n\pi_2}) \leq \dots \leq \epsilon(x^a_{n\pi_{K^a_0}})$ whatever the value of $x^a_{nk}$ is.
Suppose $\epsilon(x^a_{n\pi_m}) = B$ for some $m \in \{1,\dots,K^a_0\}$,
then by the last three statements of Theorem \ref{theorem:optimality},
we must have $\epsilon(x^a_{n\pi_l}) = 1$ for all $l < m$ and $\epsilon(x^a_{n\pi_l}) = 0$ for all $l > m$.
Therefore, once we have found the target $m$, we can fix the values of all other $x^a_{nk}$ and then obtain the value of $x^a_{n\pi_m}$
by solving $\epsilon(x^a_{n\pi_m}) = B$. 
Due to the monotonicity of $\epsilon(x^a_{n\pi_l})$ for all $l\in \{1,\dots,K^a_0\}$, the target $m$ can be searched by the bisection method.

\begin{algorithm}[t]
    \caption{Algorithm for Solving $P^c_n$}
    \label{alg:caching}
    \begin{algorithmic}[1]
        \STATE Initialization: $B_l = \epsilon_{min}, B_r = 0$;
        \WHILE{$B_r - B_l \geq \xi$}
            \STATE $B_m = \frac{B_l + B_r}{2}$;
            \STATE Solve $\bm{x}_n = g(B_m)$ using Algorithm \ref{alg:g};
            \IF{$\sum_{a\in\mathcal{A}, k\in\mathcal{K}^a} x^a_{nk} s^a_k < C^s_n$}
                \STATE $B_l = B_m$;
            \ELSE
                \STATE $B_r = B_m$;
            \ENDIF
        \ENDWHILE
        \STATE $B^* = \frac{B_l + b_r}{2}$;
        \RETURN $\bm{x}^*_n = g(B^*)$
    \end{algorithmic}
\end{algorithm}

\begin{algorithm}[t]
    \caption{Algorithm for implementing $g$}
    \label{alg:g}
    \begin{algorithmic}[1]
        \renewcommand{\algorithmicrequire}{\textbf{Input:}}
        \REQUIRE $B, \mathcal{K}^a_0, \mathcal{K}^a_1$
        \FOR{$a\in\mathcal{A}$}
            \FOR{$k\in\mathcal{K}^a_1$}
                \IF{$\epsilon(x^a_{nk}) \leq B$}
                    \STATE $x^a_{nk} = 1$;
                \ELSE
                    \STATE $x^a_{nk} = 0$;
                \ENDIF
            \ENDFOR
            \STATE Sort $k\in\mathcal{K}^a_0$ by $p^a_k/s^a_k$ and obtain $\pi_1,\dots,\pi_{K^a_0}$;
            \STATE $m_l = 1, m_r = K^a_0+1$;
            \WHILE{True}
                \STATE $m = \lfloor \frac{m_l+m_r}{2}\rfloor$;
                \STATE Let $x^a_{n\pi_l} = 1$ for all $l < m$;
                \STATE Let $x^a_{n\pi_l} = 0$ for all $l > m$;
                \IF{$\left. \epsilon(x^a_{n\pi_m})\right\vert_{x^a_{n\pi_m}=0} > B$}
                    \STATE $m_r = m$;
                \ELSIF{$\left. \epsilon(x^a_{n\pi_m})\right\vert_{x^a_{n\pi_m}=1} < B$}
                    \STATE $m_l = m$;
                \ELSE
                    \STATE Solve $x^a_{n\pi_m} = \epsilon^{-1}(B)$;
                    \STATE \textbf{break};
                \ENDIF
            \ENDWHILE
        \ENDFOR
        \RETURN $\bm{x}_n$
    \end{algorithmic}
\end{algorithm}

The detailed algorithm is summarized in Algorithm \ref{alg:caching} and \ref{alg:g}, where $\epsilon_{min}$ is the minimum
storage efficiency of all $x^a_{nk}$ and $\xi$ is the desired accuracy when searching $B^*$.
Notice that the sorting process in line $9$ of Algorithm \ref{alg:g} only needs to be executed once
and the obtained ordering $\pi_1,\dots,\pi_{K^a_0}$ can be reused next time.
Therefore, the time complexity of Algorithm \ref{alg:g} is $O(A(K^a_1+\log K^a_0))$.
As a result, the time complexity of Algorithm \ref{alg:caching} is 
$O(\log \frac{-\epsilon_{min}}{\xi} A(K^a_1+\log K^a_0))$.

The caching decision $\bm{x}^*_n$ returned by Algorithm \ref{alg:caching} is fractional but in practice, a binary solution is required.
Fortunately, by the above analysis, there is at most one fractional component in $\bm{x}^*_n$ so we can obtain a binary solution
$\hat{\bm{x}}_n$ by setting this component to $0$.
The following theorem shows that the performance of $\hat{\bm{x}}_n$ is close to the optimal one.
\begin{theorem}
    Let $\hat{\bm{x}}^*_n$ be the optimal binary caching decision, $s^{max} = \max_{a\in\mathcal{A}, k\in\mathcal{K}^a}s^a_k$
    be the maximum data size of all computation results.
    Suppose $D^0, \hat{D}$ and $\hat{D}^*$ are the weighted response time induced by $\bm{x}_n = \bm{0}, \hat{\bm{x}}_n$ 
    and $\hat{\bm{x}}_n^*$, respectively.
    Then we have
    \begin{equation*}
        \frac{D^0-\hat{D}}{D^0-\hat{D}^*} \geq 1 - \frac{s^{max} |\mathcal{A}_n|}{C^s_n},
    \end{equation*}
    where $\mathcal{A}_n = \{ a\in\mathcal{A} \, | \sum_{k\in\mathcal{K}^a} x^{a,*}_{nk} > 0\}$, i.e. 
    $|\mathcal{A}_n|$ is the number of applications that have cached results at BS $n$.
    \label{theorem:performance_gap}
\end{theorem}
\begin{proof}
   Please see Appendix \ref{appendix:3}.
\end{proof}
In practice, the storage capacity $C^s_n$ is usually much larger than the data size of a single input.
For example, the ImageNet dataset \cite{deng2009imagenet} contains $1.4$ million labeled images but only takes up $160$ GB storage space.
Therefore, the average data size of each input is around $0.1$ MB but the RAM space of edge servers usually ranges up to several GB.
In this case, the ratio $(D^0-\hat{D})/(D^0-\hat{D}^*)$ approaches $1$ thus the obtained response time $\hat{D}$ is very close to 
the optimal one $\hat{D}^*$.

Notice that the above algorithm and theoretical analysis only apply to the caching subproblem at a single BS.
To obtain a global caching decision, we can utilize the AM method again and 
apply the above algorithm to optimize the caching decisions for different BSs iteratively.


\subsection{Scheduling Subproblem} \label{subsection:scheduling}
Compared to the caching subproblem, the scheduling subproblem \eqref{subproblem:scheduling} is much easier because the decision
variables are real numbers so gradient-based methods can be used.
Since \eqref{subproblem:scheduling} is a constrained optimization problem, we can apply the Projected Gradient Descent (PGD) algorithm
to find a local minimum.
Compared to the normal gradient descent, PGD takes an extra operation that projects the updated solution to the feasible
region so that all constraints are satisfied.
In our problem, let $(\bm{\lambda}^i, \bm{f}^i)$ be the obtained solution after the $i$-th iteration step, then the next update step is
\begin{equation*}
    (\bm{\lambda}^{i+1}, \bm{f}^{i+1}) = Proj\left[ (\bm{\lambda}^i,\bm{f}^i) - \theta^i \nabla D(\bm{\lambda}^i, \bm{f}^i) \right],
\end{equation*}
where $Proj[\cdot]$ is the projection operator that maps the input to the closest point in the feasible region, $\theta^i$ is the learning rate, 
and $D(\bm{\lambda}, \bm{f}) = \sum_{a\in\mathcal{A}} \phi^a D^a$ is the objective of \eqref{subproblem:scheduling}.

To implement the projection operator, we first replace $f^a_n$ with $\hat{f}^a_n = f^a_n / C^c_n$.
Moreover, the system response time is decreasing with respect to $f^a_n$ so we can assume constraint \eqref{lmp:computing} is tight.
Therefore, constraint \eqref{lmp:computing} can be equivalently expressed as 
\begin{equation}
    \sum_{a\in\mathcal{A}} \hat{f}^a_n = 1, \quad \forall n \in\mathcal{N}.
    \label{cons:computing_eq}
\end{equation}
Combining \eqref{cons:scheduling}, \eqref{cons:computing_eq}, and \eqref{lmp:variable_ylf} we know that the feasible regions
of $(\lambda^a_n)_{n\in\mathcal{N}}$ and $(\hat{f}^a_n)_{a\in\mathcal{A}}$ are both probability simplex~\cite{krstovski2013efficient}.
According to \cite{wang2013projection}, projection to the probability simplex can be computed in $O(d\log d)$ time,
where $d$ is the dimension of the probability simplex.
Hence, in our problem, the time complexity of the projection operator is $O(NA\log(NA))$.

After implementing the projection operator, we still need to 1) decide the learning rate $\theta^i$ and 
2) guarantee the obtained solution satisfies constraint \eqref{cons:stability}.\footnote{Recall that we have relaxed this constraint previously.}
The two issues can be simultaneously addressed by the backtracking line search \cite{armijo1966minimization}.
Intuitively, the backtracking method starts with a relatively large step size and then iteratively shrinks it by a contraction
factor until a sufficient decrease of the objective function is observed.
Since our feasible region is convex, the line between $(\bm{\lambda}^i, \bm{f}^i)$ and $(\bm{\lambda}^{i+1}, \bm{f}^{i+1})$
also lies in the feasible region.
Therefore, we can first choose a relatively large $\theta^i$ and then backtracking along the direction
$\bm{p}^i = (\bm{\lambda}^{i+1}, \bm{f}^{i+1}) - (\bm{\lambda}^i, \bm{f}^i)$.
Let $\alpha\in(0,1)$ and $\beta\in(0,1)$ be the control parameters and $j=0$ be the iteration counter, 
then the standard backtracking algorithm repeatedly increments $j$ until
\begin{equation}
    D(\bm{\lambda}^i, \bm{f}^i) - D((\bm{\lambda}^i, \bm{f}^i) + \beta^j \bm{p}) \geq - \alpha \beta^j \nabla D(\bm{\lambda}^i, 
    \bm{f}^i)^T \bm{p}.
    \label{ineq:backtracking}
\end{equation}
Theoretical analysis guarantees that such $j$ always exists and is denoted by $j_1$.
In addition to the decrease of objective value, an extra requirement of our problem is that the obtained solution 
$(\bm{\lambda}^i, \bm{f}^i) + \beta^j \bm{p}$ should also satisfy constraint \eqref{cons:stability}.
Let $\mathcal{S}$ be the set of $(\bm{\lambda}, \bm{f})$ that satisfies \eqref{cons:stability}.
Notice that $\mathcal{S}$ is an open set, so there exists a neighborhood $U$ of 
$(\bm{\lambda}^i, \bm{f}^i)$ such that $U \subseteq \mathcal{S}$.
Moreover, since $\beta \in (0,1)$,
we can always find a sufficiently large $j_2$ such that for all $j \geq j_2$ we have $(\bm{\lambda}^i, \bm{f}^i) + \beta^j \bm{p} \in U$.
Let $j^* = \max(j_1, j_2)$, then $(\bm{\lambda}^i, \bm{f}^i) + \beta^{j^*} \bm{p}$ satisfies \eqref{ineq:backtracking}
and \eqref{cons:stability} simultaneously.
According to \eqref{eq:delay_no_cache}, \eqref{delay:cache} and \eqref{eq:D^a_n},
$D^a_n$ increases to infinity when $(\bm{\lambda}, \bm{f})$ approaches the boundary of $\mathcal{S}$.
Therefore, the evolution trace of $(\bm{\lambda}^i, \bm{f}^i)$ will stay away from the boundary of $\mathcal{S}$
and makes $j_2$ very small in practice.
As a result, the backtracking algorithm usually terminates in several iterations if we choose a proper $\alpha$ and $\beta$.


\section{Performance Evaluation} \label{section:simulation}

In this section, extensive numerical simulations are conducted to evaluate the effectiveness of the proposed algorithm.
Unless specified otherwise, the simulated CEC network contains $10$ BSs and $5$ applications.
We call a complete solving process of the caching subproblem and the scheduling subproblem as one round.
By default, each experiment lasts for $10$ rounds.
Recall that the algorithms solving the two subproblems are iterative and we execute both algorithms for $10$ iterations.
The rest simulation parameters are listed in Table \ref{tab:sim_param}, where
$\mathcal{U}[a,b]$ refers to the uniform distribution on $[a,b]$
and $\mathcal{N}(\mu,\sigma)$ denotes the normal distribution with mean $\mu$ and standard deviation $\sigma$.
The values of parameters are set in accordance with the empirical results in \cite{al2022promise} and \cite{al2022reservoir}.

\begin{table}[th]
\renewcommand{\arraystretch}{1.1}
\caption{Simulation Parameters}
\label{tab:sim_param}
\centering
\begin{tabularx}{\linewidth}{l|c}
\hline
\textbf{Parameter} & \textbf{Value}\\
\hline
Computational capacity $C^c_n$ & $\mathcal{U}[2,8]$ GHz \\
Storage capacity $C^s_n$ & $\mathcal{U}[2,8]$ GB \\
Number of typical inputs $K^a$ & $\mathcal{U}[10K, 50K]$ \\
Workload of each task $w^a$ & $\mathcal{U}[200M, 600M]$ CPU cycles \\
Worklaod of cache-searching $w^s$ & $25M$ CPU cycles \\
Task arrival rate $R^a_n$ & $\mathcal{U}[0.5, 1.5]$ task/second \\
Hit rate of each result $p^a_k$ & $\mathcal{U}[1.2\times 10^{-5}, 3.6\times 10^{-5}]$ \\
Communication time $D^t_n$ & $\mathcal{U}[10, 30]$ ms \\
Application weight $\phi^a$ & $1$ \\
Bisection search accuracy $\xi$ & $1\times 10^{-9}$ \\
Data size of computation result $s^a_k$ & $\mathcal{N}(0.1, 0.03)$ MB \\
Learning rate $\theta^i$ & $1/\sqrt{i}$ \\
Control parameters $\alpha, \beta$ & $\alpha = 0.3, \beta = 0.5$ \\
\hline
\end{tabularx}
\end{table}

The performance of the proposed algorithm is compared with the following benchmarks:
\begin{itemize}
    \item \emph{No Collaboration (NoC):}
        Computation tasks arrived at each BS must be processed locally and the computation results cached at one BS can
        only be used by itself.
    \item \emph{No computation Reuse (NoR):}
        Computation tasks are processed in a collaborative manner but no computation result is reused.
    \item \emph{Greedy:}
        Each BS caches the computation results with the largest $p^a_k/s^a_k$ until the storage space is used up.
        After that, the workload of each application is allocated to BSs in proportion to their computational capacities.
\end{itemize}
Notice that the initial point of the proposed algorithm is the result of the greedy algorithm.

\subsection{Convergence of the Proposed Algorithm}
\begin{figure}[t]
\centering
\subfloat[Rounds]{\includegraphics[width=0.35\textwidth]{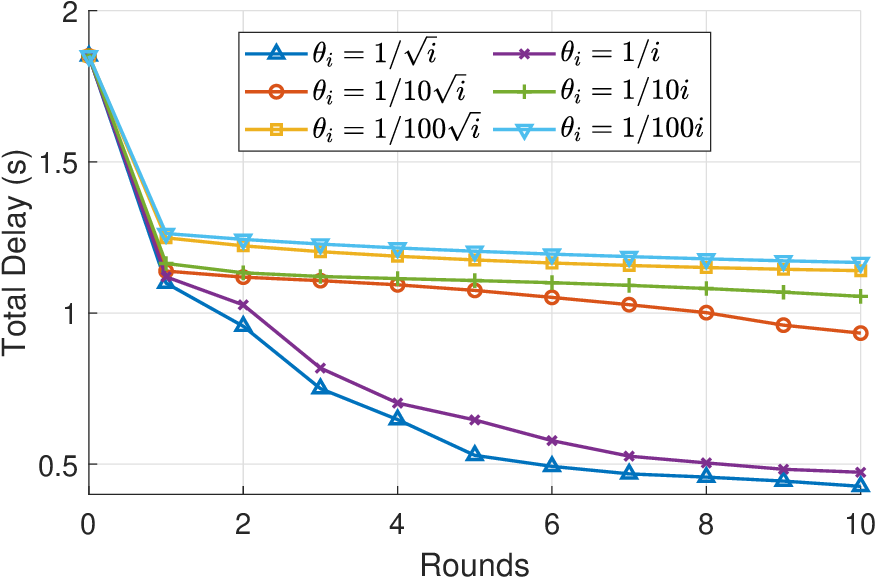} \label{fig:lr_rounds}}
\vfil
\subfloat[Iterations in the first round]{\includegraphics[width=0.35\textwidth]{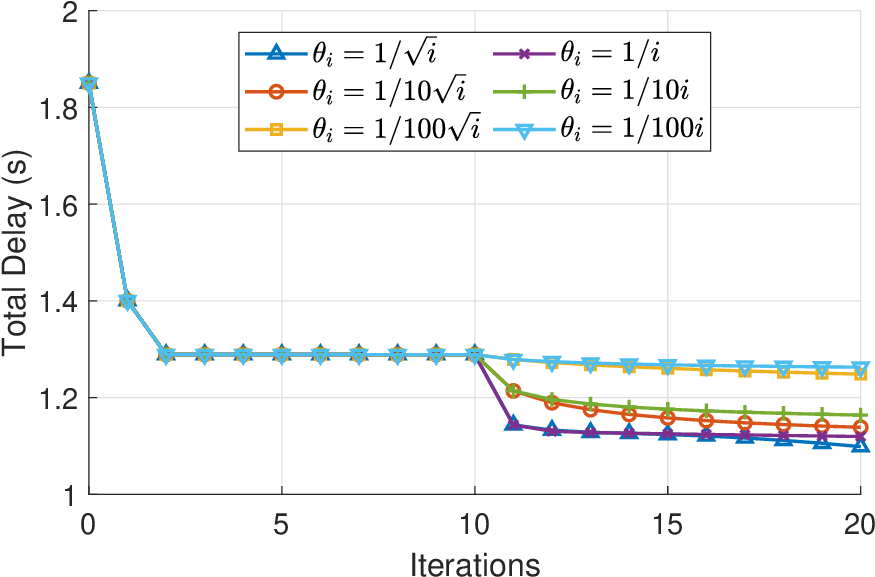} \label{fig:lr_iterations}}
\caption{Convergence of the proposed algorithm under different learning rates.}
\label{fig:lr}
\vspace{-15pt}
\end{figure}
Fig.\ref{fig:lr_rounds} demonstrates the convergence of the proposed algorithm under different rounds.
Thanks to the backtracking line search, the PGD will not oscillate even if we choose a relatively large learning rate.
Therefore, a larger learning rate also corresponds to a faster convergence speed.
To quantify the contributions of the caching and scheduling algorithms, 
Fig.\ref{fig:lr_iterations} presents the total delay under different iterations in the first round.
Recall that the caching and scheduling algorithms are executed in the first and last ten iterations, respectively.
The total delay reduces substantially at the first two caching iterations but remains the same afterward.
The influence of the learning rates is reflected in the scheduling stage.
Similar to the caching stage, the major improvement happens in the first two iterations.
Therefore, we can reduce the number of iterations in each round to achieve faster convergence.

\subsection{Impact of Workload}
\begin{figure}[t]
\centering
\includegraphics[width=0.35\textwidth]{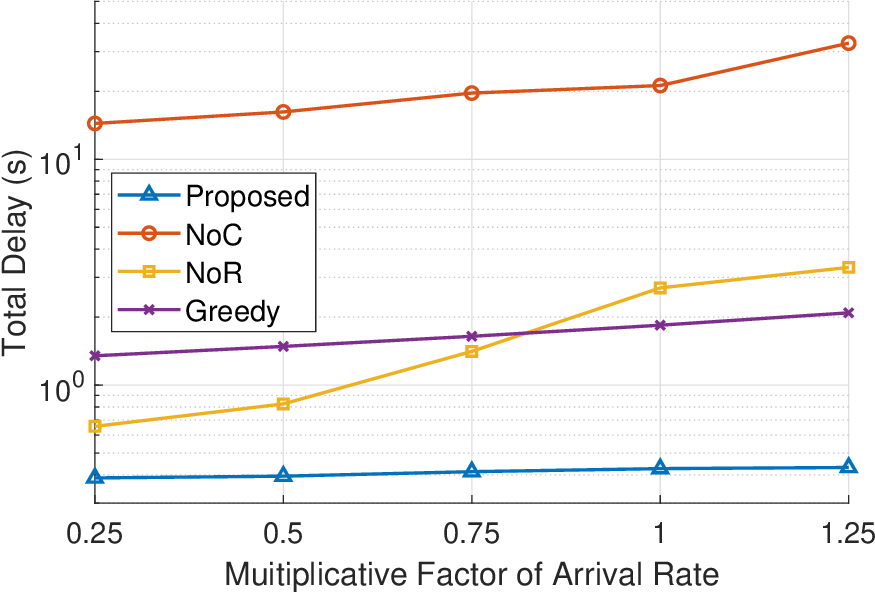}
\caption{Impact of workloads.}
\label{fig:workload}
\vspace{-15pt}
\end{figure}
To demonstrate the advantage of our algorithm,
Fig.\ref{fig:workload} shows the performance of algorithms under different workloads.
The horizontal axis represents the multiplicative factor of the task arrival rate.
For example, when the horizontal axis is $0.5$, the corresponding arrival rate is half of the default value.
In all circumstances, the proposed algorithm achieves the best delay.
When the workload is low, NoR outperforms Greedy, but the order is reversed as the workload becomes higher.
This phenomenon indicates that the benefits of computation reuse are more significant in overloaded systems.
Among the four algorithms, NoC produces a total delay that is an order of magnitude higher than the
results of the rest algorithms.
The main reason is some BSs are heavily loaded and lead to extraordinarily large delays.
This further validates our argument that the performance improvement brought by computation reuse is very limited
without collaboration among BSs.
In addition to the results shown above, we also conduct experiments with an even higher workload, but NoR
and NoC failed to meet the stability requirement of the system.
In contrast, our algorithm still performs well in these situations.

\subsection{Impact of Network Scale}
\begin{figure}[t]
\centering
\includegraphics[width=0.35\textwidth]{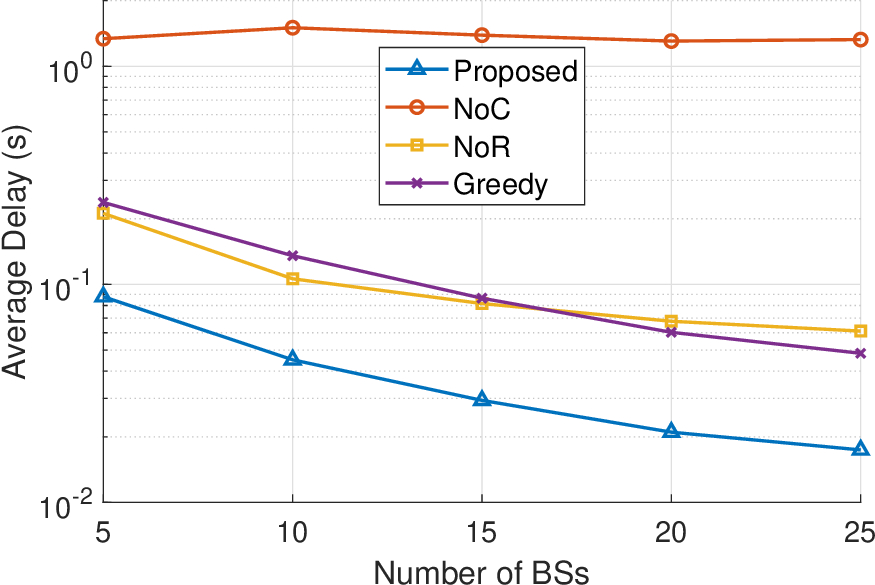}
\caption{Impact of the number of BSs.}
\label{fig:N}
\vspace{-15pt}
\end{figure}
\begin{figure}[t]
\centering
\includegraphics[width=0.35\textwidth]{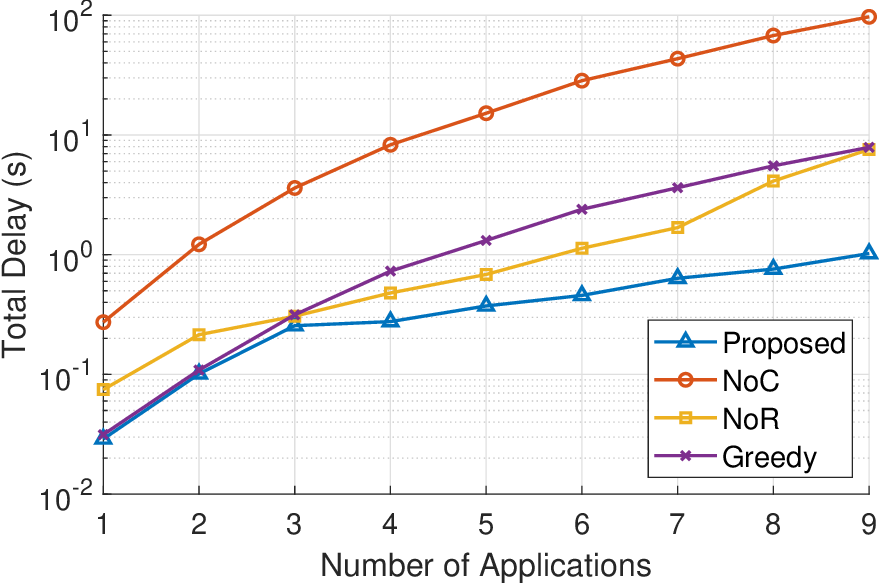}
\caption{Impact of the number of applications.}
\label{fig:A}
\vspace{-15pt}
\end{figure}
To verify the scalability of the algorithms, Fig.\ref{fig:N} presents their performance under different numbers of BSs.
To better illustrate the influence of network scale, the vertical axis is changed to the average delay,
which is the total delay dividing the number of BSs.
Unsurprisingly, the number of BSs has no effect on the average delay of NoC because each BS in NoC is an independent entity.
Since the benefits of collaboration are enhanced when more BSs joined the system,
the performance of the rest three algorithms is improved with the increase of the number of BSs.

Fig.\ref{fig:A} shows the total delay under different numbers of applications.
To exclude the influence of other factors, 
we scale down the arrival rate of each application in proportion to the number of applications
so that the total arrival rate at each BS remains unchanged across experiments.
Since the computational capacity has to be shared by multiple applications, the total delay
increases for all algorithms as the number of applications grows.
However, the proposed algorithm achieves minimal performance degradation compared to the other three.

\section{Conclusions} \label{section:conclusion}
In this paper, we have proposed a computation reuse algorithm to minimize the response time for CEC networks.
We decomposed the formulated problem into the caching subproblem and the scheduling subproblem.
Two efficient approximate algorithms have been proposed to solve the subproblems based on the necessary conditions of optimal solutions.
Numerical simulations indicate the system performance is significantly improved by applying computation reuse in CEC networks.
This result provides a new perspective to enhance the computation efficiency of edge computing systems by exploiting idle storage resources.

\appendices

\renewcommand{\thesectiondis}[2]{\Alph{section}:}

\section{Proof of Theorem \ref{theorem:np}} \label{appendix:1}
    In the considered special case, we have only one BS and one application, so there is no need to offload tasks among BSs and share computational power among applications.
    As a result, the decision variables $\lambda^a_n$ and $f^a_n$ are redundant.
    In a stable system, the total task arrival rate should be lower than the total processing rate.
    This condition reduces to constraint \eqref{cons:stability} when $N=1$ so this constraint is also redundant.
    We further assume that the benefit of cache-searching is greater than the induced cost so we have $y^a_n=1$ in the optimal solution.
    In this case, it can be easily verified that the response time is optimized when we achieve the maximum hit rate.
    Therefore, the problem \eqref{lmp} is equivalent to the following problem:
    \begin{align*}
        \max\quad &P^{a}_{hr} = P^{a,n}_{lhr} = \sum_{k\in\mathcal{K}^a} x^a_{nk}p^a_k \\
        s.t.\quad &\sum_{k\in\mathcal{K}^a} x^a_{nk}s^a_k \leq C^s_n \\
        &x^a_{nk} \in \{0,1\}, \quad \forall k\in\mathcal{K}^a.
    \end{align*}
    The above problem is a standard 0-1 knapsack problem and is known to be NP-hard.
    
\section{Proof of Theorem \ref{theorem:optimality}} \label{appendix:2}
    The first statement holds apparently because the objective function of $P^c_n$ is non-increasing with respect to $x^{a,*}_{nk}$.
    The rest three statements follow a similar argument so we only prove statement (2) for brevity.

    If there is only one component in $\bm{x}^*_n$ is fractional, then we can simply let $A$ equal its storage efficiency.
    Otherwise, we only need to show that for any two components $x^{a_1,*}_{nk_1}$ and $x^{a_2,*}_{nk_2}$ with fractional values,
    their storage efficiency is equal.
    We prove this by contradiction.
    Without loss of generality, we assume $\epsilon(x^{a_1,*}_{nk_1}) < \epsilon(x^{a_2,*}_{nk_2})$.
    Since the values of both decision variables are fractional, we can find a sufficiently small value $\delta > 0$ 
    such that $x^{a_1,*}_{nk_1} + \delta \in (0,1)$ and $x^{a_2,*}_{nk_2} - \delta \in (0,1)$.
    Then we can construct a feasible caching decision $\bm{x}'_n$ from $\bm{x}^*_n$ by replacing 
    $x^{a_1,*}_{nk_1}$ and $x^{a_2,*}_{nk_2}$ with $x^{a_1,*}_{nk_1} + \delta$ and $x^{a_2,*}_{nk_2} - \delta$.
    According to the definition of storage efficiency, the objective value resulted by $\bm{x}'_n$ is smaller than
    that of $\bm{x}^*_n$, which contradicts the assumption that $\bm{x}^*_n$ is the optimal solution.

\section{Proof of Theorem \ref{theorem:performance_gap}} \label{appendix:3}
    Let $D^*$ be the weighted response time induced by $\bm{x}^*_n$.
    Since $\bm{x}^*_n$ is the optimal solution of the relaxed problem, we must have $D^* \leq \hat{D}^*$.
    Hence we only need to prove
    \begin{equation}
        \frac{D^0-\hat{D}}{D^0-D^*} \geq 1 - \frac{s^{max} |\mathcal{A}_n|}{C^s_n}.
        \label{eq:gap}
    \end{equation}
    For convenience, we define the average storage efficiency of any caching decision at BS $n$ as
    the ratio of the decrease in response time to the allocated storage space, i.e.
    \begin{equation}
        \bar{\epsilon}(\bm{x}_n) = \frac{D^0 - D(\bm{x}_n)}{\sum_{a\in\mathcal{A}, k\in\mathcal{K}^a} x^a_{nk} s^a_k}.
        \label{eq:av_eff}
    \end{equation}
    According to our algorithm, there is at most one fractional component in $\bm{x}^*_n$ for each application
    $a\in\mathcal{A}_n$ and the storage space allocated to each fractional component is less than $s^{max}$.
    Therefore, the allocated storage space induced by $\hat{\bm{x}}_n$ satisfies 
    \begin{equation}
        \sum_{a\in\mathcal{A}, k\in\mathcal{K}^a} \hat{x}^a_{nk} > C^s_n - s^{max}|\mathcal{A}_n|.
        \label{eq:space_xhat}
    \end{equation}
    According to Theorem \ref{theorem:optimality}, the storage efficiency of integer components in $\bm{x}^*_n$
    is less than the storage efficiency of fractional components.
    Therefore, the average storage efficiency of $\hat{\bm{x}}_n$ is less than that of $\bm{x}^*_n$,
    i.e. $\bar{\epsilon}(\hat{\bm{x}}_n) \leq \bar{\epsilon}(\bm{x}^*_n)$.
    Combining \eqref{eq:av_eff}, \eqref{eq:space_xhat} and the fact that 
    $\sum_{a\in\mathcal{A}, k\in\mathcal{K}^a} x^{a,*}_{nk} = C^s_n$ results in \eqref{eq:gap}.

\bibliographystyle{IEEEtran}
\bibliography{ref}

\end{document}